\documentclass[twocolumn,twoside]{IEEEtran} 

\makeatletter
\let\NAT@parse\undefined
\makeatother

\usepackage{pstricks}
\usepackage{lipsum}
\usepackage[numbers,sort&compress]{natbib}
   
\usepackage{amssymb}
\usepackage{pgfplots}
\usepackage{graphicx}
\pgfplotsset{compat=1.3}
\usepackage[cmex10]{amsmath}
\usepackage{color}
\usepackage{flushend}

\newtheorem{corollary}{Corollary}

\newtheorem{lemma}{Lemma}
\newtheorem{definition}{Definition}

\def \M{\mathcal{M}}

\def \L {{\mathcal L }}

\def\M{ {\mathcal{M}} }

\makeatletter
\def\blfootnote{\xdef\@thefnmark{}\@footnotetext}
\makeatother

\begin{document}
\title{ On the Calculation of the Incomplete MGF\\ with Application{s to Wireless Communications}}

\author{
\vspace{3mm}
 F.J. Lopez-Martinez, J.M. Romero-Jerez and J.F. Paris}

{}
\maketitle
\begin{abstract}
\blfootnote{\noindent This work was presented in part at European Wireless 2015. F. J. Lopez-Martinez and J. F. Paris are with Departamento de Ingenier\'ia de Comunicaciones, Universidad de Malaga - Campus de Excelencia Internacional Andaluc\'ia Tech., Malaga 29071, Spain. J. M. Romero-Jerez is  with Departamento de Tecnolog\'ia Eectr\'onica, Universidad de Malaga - Campus de Excelencia Internacional Andaluc\'ia Tech., Malaga 29071, Spain. Contact e-mail: fjlopezm@ic.uma.es.\\
This work has been submitted to the IEEE for possible publication. Copyright may be transferred without notice, after which this version may no longer be accessible.

}
The incomplete moment generating function (IMGF) has paramount relevance in communication theory, since it appears in a plethora of scenarios when analyzing the performance of communication systems. We here present a general method for calculating the IMGF of any arbitrary fading distribution.  Then, we provide exact closed-form expressions for the IMGF of the very general $\kappa$-$\mu$ shadowed fading model, which includes the popular $\kappa$-$\mu$, $\eta$-$\mu$, Rician shadowed and other classical models as particular cases. We illustrate the practical applicability of this result by analyzing {several scenarios of interest in wireless communications: (1) Physical layer security in the presence of an eavesdropper, (2) Outage probability analysis with interference and background noise, (3) Channel capacity with side information at the transmitter and the receiver, and (4) Average bit-error rate with adaptive modulation}, when the fading on the desired link can be modeled by any of the aforementioned distributions.
\end{abstract}

\vspace{0mm}

\IEEEpeerreviewmaketitle

\section{Introduction}

\PARstart{T}{he} moment generating function (MGF) has played a pivotal role in communication theory for decades, as a tool for evaluating the performance of communication systems in very different scenarios \cite{Simon1998,DiRenzo2010,Yilmaz2012,Rao2015}. The MGF of the signal-to-noise ratio (SNR) $\gamma$, defined as the Laplace transform of the Probability Density Function (PDF) of $\gamma$, is well-known for most popular fading distributions \cite{AlouiniBook,Ermolova2008,Paris2014} and hence enables for a simple characterization of the performance metrics of interest in closed-form.

A more general function extensively used in communication theory is the so-called incomplete MGF (IMGF), {also referred to as truncated MGF, or interval MGF}. This function has an additional degree of freedom by allowing the lower (or equivalently, upper) limit of the integral in the Laplace transform, say $\zeta$, to be greater than zero (or equivalently, lower than $\infty$), and also appears when characterizing the performance in a number of scenarios of interest: {order statistics \cite{Nuttall2002,Nam2010,Yang2011,Nam2014},} symbol and bit error rate calculation with multiantenna reception \cite{Xiaodi2006,Loskot2009}, capacity analysis in fading channels \cite{DiRenzo2010}, outage probability analysis in cellular systems \cite{Annamalai2001,Morales2012}, adaptive scheduling techniques \cite{Gaaloul2012}, cognitive relay networks \cite{Yang2015} or physical layer security \cite{JuanmaJavi15}. Despite its usefulness, to the best of our knowledge the expressions of the IMGF are largely unknown for fading distributions other than the classical ones: Rayleigh, Rice, Nakagami-$m$ and Hoyt. {In fact, a general and systematic way to find analytical expressions for the IMGF does not yet exist, thus requiring the use of state-of-the-art numerical techniques for its evaluation \cite{DiRenzo2010}.}


Thence, there is a twofold motivation for this paper from a purely communication-theoretic perspective: first, we present a {general theory} for deriving the IMGF of the SNR for an arbitrary distribution. {Specifically, we show that the IMGF of \textit{any} fading distribution is given in terms of an inverse Laplace transform of a shifted version of the MGF scaled by the Laplace domain variable. This implies that the IMGF should have a similar functional form as the cumulative distribution function (CDF).} { Secondly, we exemplify the usefulness of this result by deriving} a closed-form expression for the IMGF of the $\kappa$-$\mu$ shadowed distribution \cite{Paris2014,Cotton2015}, which includes popular fading distributions such as \mbox{$\kappa$-$\mu$}, $\eta$-$\mu$ \cite{Yacoub07,Laureano2015} or Rician-shadowed \cite{Abdi2003} as particular cases, and for all of which the IMGF had not been previously reported. 

In order to illustrate its {practical applicability, we introduce several scenarios of interest in wireless communications: As a main application,} we focus on a physical-layer security set-up on which two legitimate peers (Alice and Bob) wish to communicate in the presence of an external eavesdropper (Eve). The characterization of the maximum rate at which a secure communication can be attained, i.e. the \emph{secrecy capacity} $C_S$, is a classical problem \cite{Wyner1975,Cheong1978} in communication theory. Remarkably, the research on physical-layer security of communication systems operating in the presence of fading has been boosted in the last years ever since the original works in \cite{Barros2006,Bloch2008}. The fact that Eve and Bob observe independent fading realizations adds an additional layer of security to communication compared to the conventional set-up for the Gaussian wiretap channel \cite{Cheong1978}. Hence, a secure communication is feasible even when the average SNR at Bob is lower than the average SNR at Eve. 

In the literature, there is a great interest on understanding how the consideration of more sophisticated fading models than Rayleigh may impact the secrecy performance attending to different metrics: the outage probability of secrecy capacity (OPSC), the probability of strictly-positive secrecy capacity (SPSC) or the average secrecy capacity (ASC). Specifically, fading models such as Nakagami-$m$ \cite{Sarkar09}, Rician \cite{Liu13}, Weibull \cite{Liu13b}, two-wave with diffuse power \cite{Wang14}, Nakagami-$q$ \cite{Romero2014} or $\kappa$-$\mu$ \cite{Bhargav2016} have been considered. However, the statistical characterization of the OPSC in a tractable form is often unfeasible due to the involved mathematical derivations. Thus, approximations are usually required for evaluating the OPSC \cite{Wang14,Bhargav2016}, being only possible the calculation of the SPSC in closed-form \cite{Liu13,Liu13b,Bhargav2016}. 

Very recent works have proposed novel approaches to deriving the secrecy performance metrics in a general way: in \cite{Gomez2015}, a duality between the OPSC and the outage probability analysis in the presence of interference and background noise was presented, which greatly simplifies the analysis of the OPSC calculation for an arbitrary choice of fading distributions for the desired and eavesdropper links. In \cite{Peppas2016}, a unified MGF approach to the analysis of physical-layer security in wireless systems was introduced, allowing for a numerical evaluation of the secrecy performance metrics for arbitrary fading distributions.

{As we will later see, the practical application of our results allows that the OPSC can be evaluated directly by specializing the IMGF of the SNR at Bob at some specific values}. Thus, the OPSC (and hence the SPSC as a special case) can be evaluated in closed-form provided that such IMGF is given in closed-form. This is exemplified for the very general case of the $\kappa$-$\mu$ shadowed distribution \cite{Paris2014}.

{We also illustrate how the IMGF can be used to obtain other performance metrics in relevant scenarios in communication theory which, despite looking rather dissimilar at a first glance, they all require for the computation of the IMGF. The first one of these additional scenarios is the outage probability analysis of wireless communications systems affected by interference and background noise. As previously stated, this problem was recently shown to be mathematically equivalent to the OPSC in \cite{Gomez2015}; thus, the outage probability in this scenario will also be expressed in terms of the IMGF under the same conditions assumed when computing the OPSC. The second additional scenario is related to the analysis of the channel capacity when side information is available at both the transmitter and the receiver sides \cite{Goldsmith1997}. According to the framework introduced in \cite{DiRenzo2010}, the capacity in this scenario can be expressed in terms of the IMGF. Finally, we also analyze a classical performance metric in communication theory, which is the average bit-error rate (ABER) with adaptive modulation \cite{Gold05}. The ABER in this scenario under \textit{arbitrary} fading is also given in terms of the IMGF of the fading distribution.}

The remainder of this paper is structured as follows: in Section \ref{mathematics}, the main mathematical contributions of this paper are presented, within the most relevant is a general way to deriving the IMGF of any arbitrary fading distribution. Closed-form expressions for the IMGF of the $\kappa$-$\mu$ shadowed distribution are also given, as well as for all the special cases included therein ($\kappa$-$\mu$, $\eta$-$\mu$, Rician shadowed, Rician, Nakagami-$m$, Nakagami-$q$, Rayleigh and one-sided Gaussian). Then, in Section \ref{analysis} these mathematical results are used to present an IMGF-based approach to the physical-layer security analysis in wireless systems. {Section \ref{analysis2} is devoted to illustrate how additional scenarios of interest in wireless communications can also be analyzed by using the IMGF.} Numerical results are given in Section \ref{numerical}, whereas the main conclusions are outlined in Section \ref{conclusions}.

\section{Mathematical Results}
\label{mathematics}


\begin{definition}
[Lower IMGF]
Let $X$ be a non-negative random variable and $\zeta$ a non-negative real number,
the lower IMGF of $X$ is defined as
\label{IMGF}
\begin{equation}
\label{def1}
\M_{X}^{l}(s,\zeta)
 \triangleq \int_{0}^{\zeta}  {e^{sx} f_X \left( x \right)dx}.
\end{equation}
\end{definition}
\vspace{5mm}
\begin{definition}
[Upper IMGF]
Let $X$ be a non-negative random variable and $\zeta$ a non-negative real number,
the upper IMGF of $X$ is defined as
\label{IMGFb}
\begin{equation}
\label{def1b}
\M_{X}^{u}(s,\zeta)
 \triangleq \int_{\zeta}^\infty  {e^{sx} f_X \left( x \right)dx}.
\end{equation}
\end{definition}
\vspace{5mm}
Obviously, the MGF of $X$ is obtained from these IMGFs as $\M_{X}(s)=\M_{X}^{l}(s,\infty)=\M_{X}^{u}(s,0)$. Both incomplete IMGFs are easily related through the following equation
\begin{align}
\label{obvious}
{\M}_X^u \left( {s,\zeta} \right) &= {\M}_X \left( s \right) - {\M}_X^l \left( {s,\zeta} \right),
\end{align}
In the following Lemma, we present a general expression for the IMGF of an arbitrarily distributed non-negative random variable.
\begin{lemma}
\label{lema1}
Let $X$ be a non-negative random variable with MGF $\M_{\gamma}(s)$. Its lower IMGFs
can be computed by the inverse Laplace transform of the scaled-shifted MGF, i.e.
\begin{align}
\label{eq_lema1}
{\M}_X^l \left( {s,\zeta} \right) &= {\L}^{ - 1} \left\{ {\frac{1}
{p}{\M}_X \left( {s - p} \right);p,z=\zeta} \right\},
\end{align}
where ${\L}\left\{ {h\left( z \right);z,p} \right\}
\triangleq \int_0^\infty  {e^{ - pz} h\left( z \right)dz}$ represents the Laplace transform
from the $z$-domain to the $p$-domain, and ${\L}^{ - 1} \left\{ {H\left( z \right);p,z} \right\}$ the
inverse Laplace transform from the \mbox{$p$-domain} to the $z$-domain.
\end{lemma}

\begin{proof}
Let us consider the upper IMGF in (\ref{def1}) as a function of the upper integration limit, i.e.
\begin{equation}
\label{eq_lema1_proof_2}
\Lambda(z)=\int_0^{z} {e^{sx} f_X \left( x \right)dx}.
\end{equation}
The Laplace transform of $\Lambda(z)$ in the $p$-domain can be expressed as
\begin{equation}
\label{eq_lema1_proof_3}
\begin{gathered}
  {\L}\left\{\Lambda \left( z \right);z,p\right\} = {\L}\left\{ {\int_0^z {e^{sx} f_X \left( x \right)dx;z,p} } \right\} = \frac{1}
{p}{\L}\left\{ {e^{sz} f_X \left( z \right);z,p} \right\} \hfill \\
   = \frac{1}
{p}{\L}\left\{ {f_X \left( x \right);z,p - s} \right\} = \frac{1}
{p}{\M}_X \left( {s - p} \right), \hfill \\
\end{gathered}
\end{equation}
where both the definite integral property and the modulation property of the Laplace transform have been applied in order to complete the proof.
\end{proof}
\vspace{4mm}

This Lemma provides a general way of computing the IMGF {of \textit{any} fading distribution} in terms of the conventional MGF. Interestingly, (\ref{eq_lema1}) involves an inverse Laplace transform of a shifted version of the MGF, scaled by the Laplace domain variable $p$. Thus, it is expectable that the IMGF has a functional form similar to the CDF, since the CDF arises as a particular case of the lower IMGF when evaluated in ${s=0}$, i.e.
\begin{equation}
F_X(\zeta)=\M_{X}^{l}(0,\zeta).
\end{equation}

{This result also suggests that the IMGF of \textit{any} distribution for which either the CDF or the MGF are not available in closed-form, will not be likely to have a closed-form expression. Otherwise, we would be finding an expression for such CDF or MGF as a special case. This is the case of some relevant fading distributions in the literature like Durgin and Rappaport's Two-Wave with Diffuse Power (TWDP) fading model \cite{Durgin2002}, or the Beckmann distribution \cite{Beckmann1962,Beckmann1964}. The CDF for these distributions is only available in integral form, but their MGFs have a closed-form expression \cite{AlouiniBook,Rao2015}. Another valid example is the Lognormal distribution, for which the CDF has a closed-form expression in terms of the Gaussian $Q$-function, but its MGF has been largely unknown \cite{Tellambura2010}.}

As consequences of Lemma \ref{lema1}, the IMGF of any non-negative RV can be obtained in a very general form by an inverse Laplace transform operation over the MGF. As we will now show, this Laplace transform can be computed in closed-form for a number of fading distributions of interest. In the following corollaries (and summarized in Table \ref{table01}), we derive closed-form expressions for the IMGF of the $\kappa$-$\mu$ shadowed distribution and the special cases included therein, namely the Rician shadowed, $\kappa$-$\mu$ and $\eta$-$\mu$ distributions. All these expressions are new in the literature to the best of our knowledge.

\vspace{2mm}
\begin{corollary}
\label{corolary1}
Let $\gamma$ be a $\kappa$-$\mu$ shadowed random variable with $E\left\{\gamma\right\}=\bar\gamma$, and non-negative
real shape parameters $\kappa$, $\mu$ and $m$, i.e, ${\gamma\sim \mathcal{S}_{\kappa\mu m}(\gamma;\kappa,\mu,m)}$\footnote{The symbol $\sim$ reads as \textit{statistically distributed as.}} \cite{Paris2014}. Then, its lower IMGF is given in the first entry of Table \ref{table01},
where $\Phi_2(\cdot)$ is the bivariate confluent hypergeometric function
defined in \cite[eq. 9.261.2]{Gradstein2007}.
\end{corollary}

\begin{proof}
For the sake of clarity let us write the MGF of the $\kappa$-$\mu$ shadowed
distribution as follows
\begin{equation}
\label{eq_corolary1_proof1}
\begin{gathered}
  {\M}_{\gamma} \left( s \right) = ( - 1)^\mu  A(s - a)^{m - \mu } (s - b)^{ - m}  \hfill \\
  \left\{ \begin{gathered}
  A = \frac{{\mu ^\mu  m^m \left( {1 + \kappa } \right)^\mu  }}
{{\bar \gamma ^\mu  \left( {\mu \kappa  + m} \right)^m }} \hfill \\
  a = \frac{{\mu \left( {1 + \kappa } \right)}}
{{\bar \gamma }} \hfill \\
  b = \frac{{\mu \left( {1 + \kappa } \right)}}
{{\bar \gamma }}\frac{m}
{{\mu \kappa  + m}} \hfill \\
\end{gathered}  \right. \hfill \\
\end{gathered}.
\end{equation}

In order to apply Lemma \ref{lema1}, the scaled-shifted MGF of $\gamma$ is
first expressed as follows
\begin{equation}
\label{eq_corolary1_proof2}
\begin{gathered}
  \frac{1}
{p}{\M}_{\gamma} \left( {s - p} \right) = \frac{{A( - 1)^\mu  }}
{p}(s - p - a)^{m - \mu } (s - p - b)^{ - m}  \hfill \\
   = \frac{A}
{p}(p - s + a)^{m - \mu } (p - s + b)^{ - m}  \hfill \\
   = \frac{A}
{{p^{1 + \mu } }}\left( {1 - \left( {\frac{{s - a}}
{p}} \right)} \right)^{m - \mu } \left( {1 - \left( {\frac{{s - b}}
{p}} \right)} \right)^{ - m}  \hfill \\
   = \frac{A}
{{p^{1 + \mu } }}\left( {1 - \left( {\frac{{s - a}}
{p}} \right)} \right)^{ - (\mu  - m)} \left( {1 - \left( {\frac{{s - b}}
{p}} \right)} \right)^{ - m}.  \hfill \\
\end{gathered}
\end{equation}

To perform the inverse Laplace transform the following pair
can be considered \cite[pp. 290]{Srivastava1985}
\begin{equation}
\label{eq_corolary1_proof3}
\begin{gathered}
  {\L}\left\{ {x^{\gamma  - 1} \Phi _2^{\left( n \right)} \left( {\beta _1 ,...,\beta _n ;\gamma ;\lambda _1 x,...,\lambda _n x} \right);x,p} \right\} =  \hfill \\
   = \frac{{\Gamma \left( \gamma  \right)}}
{{p^\gamma  }}\left( {1 - \frac{{\lambda _1 }}
{p}} \right)^{ - \beta _1 } ...\left( {1 - \frac{{\lambda _n }}
{p}} \right)^{ - \beta _n } ,\quad  \hfill \\
  \operatorname{Re} \left\{ \gamma  \right\} > 0,\operatorname{Re} \left\{ p \right\} > \max \left\{ {0,\operatorname{Re} \left\{ {\lambda _1 } \right\},
  ...,\operatorname{Re} \left\{ {\lambda _n } \right\}} \right\}. \hfill \\
\end{gathered}
\end{equation}
Joining (\ref{eq_corolary1_proof2}), (\ref{eq_corolary1_proof3}) and Lemma \ref{lema1} the proof is completed.
\end{proof}

\vspace{2mm}
It must be noted that the IMGF here obtained has the same functional form as the CDF given in \cite[eq. 6]{Paris2014}; thus, evaluating the IMGF has the same complexity as evaluating the CDF.
\vspace{2mm}

\begin{table*}[h]
  \renewcommand{\arraystretch}{3}
\centering
\caption{Lower IMGFs for the $\kappa$-$\mu$ shadowed fading model and particular cases included therein. Upper IMGFs can be obtained using the MGFs given in \cite{Abdi2003,Ermolova2008,Paris2014} and restated in the table, and then using (\ref{obvious}).}
\label{table01}
\begin{tabular}{|c|c|}
\hline\hline
Fading model & IMGF $\M_{\gamma}^{l}(s,z)$ and MGF $\M_{\gamma}(s)$ \\ \hline\hline
$\kappa$-$\mu$ shadowed & $ \M_{\gamma}^{l}(s,z) =  \frac{{\mu ^{\mu-1}  m^m \left( {1 + \kappa } \right)^\mu  }}
{{\Gamma(\mu)\bar \gamma ^\mu  \left( {\mu \kappa  + m} \right)^m }}z^\mu   \times 
  \Phi _2 \left( {\mu  - m,m;1 + \mu ;\left( {s - \tfrac{{\mu \left( {1 + \kappa } \right)}}
{{\bar \gamma }}} \right)z,\left( {s - \tfrac{{\mu \left( {1 + \kappa } \right)}}
{{\bar \gamma }}\tfrac{m}
{{\mu \kappa  + m}}} \right)z} \right)$ \\ & $ \M_{\gamma}(s) =  \frac{( - \mu)^\mu{ m^m \left( {1 + \kappa } \right)^\mu  }}
{{\bar \gamma ^\mu  \left( {\mu \kappa  + m} \right)^m }}\frac{{\left( {s - \frac{{\mu \left( {1 + \kappa } \right)}}
{{\bar \gamma }}} \right)^{m - \mu } }}
{{\left( {s - \frac{{\mu \left( {1 + \kappa } \right)}}
{{\bar \gamma }}\frac{m}
{{\mu \kappa  + m}}} \right)^m }}$  \\ \hline
 Rician shadowed &  $ \M_{\gamma}^{l}(s,z) =  \frac{ m^m \left( {1 + K } \right) }{{\bar \gamma \left( {K  + m} \right)^m }} z \times 
  \Phi _2 \left( {1  - m,m;2 ;\left( {s - \tfrac{{ \left( {1 + K } \right)}}
{{\bar \gamma }}} \right)z,\left( {s - \tfrac{{ \left( {1 + K } \right)}}
{{\bar \gamma }}\tfrac{m}
{{K  + m}}} \right)z} \right)$  \\ & $ \M_{\gamma}(s) =  - \frac{{ m^m \left( {1 + K } \right) }}
{{\bar \gamma  \left( { K  + m} \right)^m }}\frac{{\left( {s - \frac{{ \left( {1 + K } \right)}}
{{\bar \gamma }}} \right)^{m - 1 } }}
{{\left( {s - \frac{{ \left( {1 + K } \right)}}
{{\bar \gamma }}\frac{m}
{{ K  + m}}} \right)^m }}$  \\ \hline
$\kappa$-$\mu$ &  $\M_{\gamma}^{l}(s,z) = \frac{{{\mu ^\mu }{{\left( {1 + \kappa } \right)}^\mu }}}{{{{\left( {\mu
\left( {1 + \kappa } \right) - \bar \gamma s} \right)}^\mu }}}\exp \left(
{\frac{{\mu \kappa \bar \gamma s}}{{\mu \left( {1 + \kappa } \right) - \bar
\gamma s}}} \right) \times
\left[ {{\rm{1 - }}{Q_\mu }\left( {\sqrt {2\mu \kappa \frac{{\mu \left( {1
+ \kappa } \right)}}{{\mu \left( {1 + \kappa } \right) - \bar \gamma s}}}
,\sqrt {2\left( {\frac{{\mu \left( {1 + \kappa } \right)}}{{\bar \gamma }}
- s} \right)z} } \right)} \right]$
  \\ & $ \M_{\gamma}(s) =  \frac{{\mu ^\mu  \left( {1 + \kappa } \right)^\mu  }}
{{\left( {\mu \left( {1 + \kappa } \right) - \bar \gamma s} \right)^\mu  }}\exp \left( {\frac{{\mu \kappa \bar \gamma s}}
{{\mu \left( {1 + \kappa } \right) - \bar \gamma s}}} \right)$ \\ \hline
$\eta$-$\mu$ & $ \M_{\gamma}^{l}(s,z) =   \frac{\mu ^{2\mu-1} }
{2{\Gamma(2\mu)\bar \gamma ^{2\mu}}} \left(\frac{\left( {1 + \eta } \right)^2}{  \eta} \right)^{\mu} z^{2\mu}   \times 
  \Phi _2 \left( \mu,\mu;2\mu+1 ;\left( {s - \tfrac{{\mu \left( {1 + \eta } \right)}}
{\eta{\bar \gamma }}} \right)z,\left( {s - \tfrac{{\mu \left( {1 + \eta } \right)}}
{{\bar \gamma }}}\right) z\right)$ \\ & $ \M_{\gamma}(s) =  \left(\frac{\mu^2\left(2+\eta^{-1}+\eta\right)}{\left[(1+\eta)\mu-\bar\gamma s\right]\left[(1+\eta^{-1})\mu-\bar\gamma s\right]}\right)^{\mu}$  \\ \hline
\end{tabular}
\end{table*}


\begin{corollary}
\label{corolary2a}
Let $\gamma$ be a Rician shadowed random variable with $E\left\{\gamma\right\}=\bar\gamma$, and non-negative
real shape parameters $K$ and $m$, i.e, ${\gamma\sim \mathcal{S}_{Km}(\gamma;K,m)}$ \cite{Abdi2003}. Then, its lower IMGF is given in the second entry of Table \ref{table01}.
\end{corollary}
\begin{proof}
Specializing the results for the $\kappa$-$\mu$ shadowed distribution for $\mu=1$, the IMGF for the Rician-shadowed case is obtained with $K=\kappa$.
\end{proof}

\vspace{2mm}
Again, the IMGF of the Rician Shadowed distribution has the same functional form as its CDF \cite{Paris2010}. We must note that for the case of $m\in\mathbb{Z}$, we have that $\Phi_2(1-m,m;2;\cdot,\cdot)$ function reduces to a finite sum of Laguerre polynomials \cite[eq. 4]{Paris2010}; similarly, for $m$ being a positive half integer, then $\Phi_2(1-m,m;2;\cdot,\cdot)$ function reduces to a finite sum of Kummer hypergeometric functions, modified Bessel functions and Marcum-$Q$ functions \cite[eq. 7]{Paris2010}.
\vspace{2mm}

\begin{corollary}
\label{corolary2}
Let $\gamma$ be a $\kappa$-$\mu$ random
variable with ${E\left\{\gamma\right\}=\bar\gamma}$ and non-negative
real shape parameters $\kappa$ and $\mu$, i.e, ${\gamma\sim \mathcal{S}_{\kappa\mu}(\gamma;\kappa,\mu)}$ \cite{Yacoub07}. Then, its lower IMGF is given in the third entry of Table \ref{table01}, where $Q_{\mu}$ is the generalized Marcum $Q$ function of $\mu$-th order \cite{AlouiniBook}.
\end{corollary}

\begin{proof}
Using the MGF of a $\kappa$-$\mu$ distributed RV $\gamma$ in \cite{Ermolova2008}, and according to Lemma \ref{lema1}, the lower IMGF can be expressed as
\begin{align}
\label{eq_corolary2_proof1}
&{\M}_{\gamma}^l \left( {s,z} \right) = \L^{-1}\left\{\tfrac{1}{p}\M_{\gamma}(s-p);p,z\right\}\nonumber\\&={\L^{ - 1}}\left\{ {\tfrac{1}{p}\tfrac{{{\mu ^\mu }{{\left( {1 + \kappa }
\right)}^\mu }}}{{{{\left( {\mu \left( {1 + \kappa } \right) - \bar \gamma
s + \bar \gamma p} \right)}^\mu }}}\exp \left( {\tfrac{{\mu \kappa \bar
\gamma s - \mu \kappa \bar \gamma p}}{{\mu \left( {1 + \kappa } \right) -
\bar \gamma s + \bar \gamma p}}} \right);p,z} \right\},
\end{align}
After some algebra, the terms in (\ref{eq_corolary2_proof1}) can be conveniently rearranged so the inverse Laplace transform has the following form:
\begin{align}
\label{eq_corolary2_proof2}
\begin{array}{l}
{\M}_{\gamma}^l \left( {s,z} \right) = \exp \left( { - \mu \kappa } \right){\left( {\frac{{\mu \left( {1 + \kappa
} \right)}}{{\bar \gamma }}} \right)^\mu }\exp \left( {\left( {s -
\frac{{\mu \left( {1 + \kappa } \right)}}{{\bar \gamma }}} \right)z}
\right) \times \\
{\L^{ - 1}}\left\{ {\frac{1}{{{p^{\mu  + 1}}}}\frac{1}{{1 - \left(
{\frac{{\mu \left( {1 + \kappa } \right)}}{{\bar \gamma }} - s}
\right)/p}}\exp \left( {\frac{1}{p}{{\frac{{{\mu ^2}\kappa \left( {1 + \kappa }
\right)}}{{\bar \gamma }}}}} \right);p,z} \right\},
\end{array}
\end{align}
This inverse Laplace transform can be calculated in terms of the bivariate confluent hypergeometric function $\Phi_3(\cdot)$ defined in \cite[eq. 9.261.3]{Gradstein2007} by using the transform pair in \cite[4.24.9]{Erdelyi1954}. This yields 
\begin{align}
\label{eqPepe}
\begin{array}{l} {\M}_{\gamma}^l \left( {s,z} \right) =  \frac{\mu^{\mu}(1+\kappa)^{\mu}z^{\mu}}{\Gamma(\mu+1)\bar\gamma^{\mu}\exp{(\kappa\mu)}}
\exp \left( {\left( {s -
\frac{{\mu \left( {1 + \kappa } \right)}}{{\bar \gamma }}} \right)z}
\right) \\ \times {\Phi _3}\left(
{1,\mu  + 1;\left( {\frac{{\mu \left( {1 + \kappa } \right)}}{{\bar \gamma
}} - s} \right)z,\frac{{{\mu ^2}\kappa \left( {1 + \kappa } \right)}}{{\bar
\gamma }}z} \right),
\end{array}
\end{align}

Finally, we make use of the existing connection between the $\Phi_3(1,\mu+1;a;b)$ function and the generalized Marcum $Q$-function, which holds for $\mu\in\mathbb{R}$. Adapting \cite[eq. 34]{Morales2014} to the specific case here addressed, we have
\begin{align}
\frac{\Phi_3\left(1,\mu+1;a,b\right)}{\Gamma(\mu+1)}=\exp\left(a+\tfrac{b}{a}\right)a^{-\mu}\left[1-Q_{\mu}\left(\sqrt{2\tfrac{b}{a}},\sqrt{2a}\right)\right].
\end{align}
Combining this equation with (\ref{eqPepe}) yields the final expression given in the third entry of \mbox{Table \ref{table01}}. This completes the proof.
\end{proof}

\vspace{2mm}
Note that the IMGF is given in terms of the well-known Marcum $Q$-function, just like the CDF of the $\kappa$-$\mu$ distribution originally derived by Yacoub \cite{Yacoub07}.

\vspace{2mm}

\begin{corollary}
\label{corolary3}
Let $\gamma$ be a $\eta$-$\mu$ random
variable with ${E\left\{\gamma\right\}=\bar\gamma}$ and non-negative
real shape parameters $\eta$ and $\mu$, i.e, ${\gamma\sim \mathcal{S}_{\eta\mu}(\gamma;\eta,\mu)}$ \cite{Yacoub07}. Then, its lower IMGF is given in the fourth entry of Table \ref{table01}.
\end{corollary}

\begin{proof}
Leveraging the recent connection between the $\kappa$-$\mu$ shadowed distribution and the $\eta$-$\mu$ distribution \cite{Laureano2015}, the IMGF of the $\eta$-$\mu$ power envelope in format 1 is obtained\footnote{According to the original definition in \cite{Yacoub07}, this format implies that $\eta\in[0,\infty)$} by setting the parameters of the $\kappa$-$\underline\mu$ shadowed distribution to $\underline \mu = 2\mu$, $\kappa=\frac{1-\eta}{2\eta}$ and $m=\mu$.
\end{proof}

\vspace{2mm}
This expression is coincident with the one obtained in \cite{Morales2014}, and also has the same complexity as the original CDF for the $\eta$-$\mu$ distribution derived in \cite{Morales2010} .
\vspace{2mm}

\section{Application to Physical Layer Security}
\label{analysis}
\subsection{Single-antenna scenario}
Aided by the previous mathematical results, we will now show how the IMGF can be used to directly analyze the performance in wireless scenarios. Specifically, we aim to determine the physical layer security of a wireless link between two legitimate peers (Alice and Bob) in the presence of an external eavesdropper (Eve) \cite{Barros2006,Bloch2008}. 
We first consider that Bob, Eve and Alice are equipped with single-antenna devices; the inclusion of multiple antennas at Bob or Eve will be later discussed in this section. Depending on the characteristics of the propagation environment, random fluctuations affecting the desired and wiretap links need to be modeled with a specific distribution.



The performance in this scenario can be characterized by the secrecy capacity $C_S$, defined as
\begin{equation}
\label{sec01}
C_S\triangleq C_b-C_e = \log_2\left({\frac{1+\gamma_b}{1+\gamma_e}}\right)>0,
\end{equation}
where $\gamma_b$ and $\gamma_e$ are the instantaneous SNRs at Bob and Eve, respectively, and $C_b$ and $C_e$ denote the capacities of the communication links between Alice and Bob, and between Alice and Eve, respectively. 



In many cases, Eve's channel state information (CSI) is not available at Alice and hence {information-theoretic security cannot be guaranteed}. This may be the case on which Eve is a passive eavesdropper, and hence Alice has no way to have access to Eve's CSI. Conversely, perfect knowledge of Bob's CSI by Alice can be assumed. Thus, Alice selects a constant secrecy rate $R_S$ for transmission; in this situation, the outage probability of the secrecy capacity (OPSC) gives the probability that communication at a certain secrecy rate ${R_S > 0}$ cannot be securely attained. This metric is computed as follows
\begin{align}
\label{eq_analysis_1}
\mathcal{P}_{R_S}\triangleq\Pr \left\{ {C_S  \leqslant R_S } \right\}&=\Pr\left\{ {\gamma_b  \leqslant \left(2^{R_S}-1\right)\left(\frac{2^{R_S}}{2^{R_S}-1}\gamma_e+1\right) } \right\}\\
&= 1 - \Pr \left\{ {C_S  > R_S } \right\}.
\end{align}
The probability of strictly positive secrecy capacity can be obtained as a particular case of (\ref{eq_analysis_1}) by setting ${R_S=0}$.

Finally, another secrecy performance metric of interest is the $\epsilon$-outage secrecy capacity $C_{\epsilon}$. This is defined as the largest secrecy rate $R_S$ for which the OPSC satisfies ($\mathcal{P}_{R_S}{\leq\epsilon}$), with $0\leq\epsilon\leq 1$. For a certain $\epsilon$, this metric is computed as
\begin{equation}
C_{\epsilon}\triangleq \underset{R_S:\mathcal{P}_{R_S}\leq\epsilon}{\sup}  \left\{R_S\right\},
\end{equation}

{We will first assume that the fading experienced by the eavesdropper can be modeled by the $\kappa$-$\mu$ shadowed distribution, i.e. ${\gamma_e\sim \mathcal{S}_{\kappa\mu m}(\gamma_e;\kappa,\mu,m)}$. This distribution \cite{Paris2014} is well suited to model both line-of sight (LOS) and non-LOS (NLOS) scenarios, and also includes most popular fading distributions in the literature as special cases. We will also consider that the fading severity parameters $\mu$ and $m$ take integer values\footnote{This can be justified as follows: the parameter $\mu$ in the $\kappa$-$\mu$ distribution introduced by Yacoub \cite{Yacoub07} was defined as the number of clusters of multipath waves propagating in a certain environment; thus, according to this definition the consideration of integer $\mu$ is related to the physical model for the $\kappa$-$\mu$ distribution. Equivalently, the restriction of $m$ to take integer values does not have a major impact unless the LOS component is affected by heavy shadowing (i.e. very low values of $m$). In practice, this restriction has a negligible effect, and specially when the $\kappa$-$\mu$ shadowed distribution is used to approximate the $\kappa$-$\mu$ distribution in a more tractable form \cite{Lopez2016}.}, which allows for a simpler mathematical tractability. 


With this only restriction, and for \textit{any arbitrary} fading distribution at the legitimate link between Alice and Bob, the OPSC can be computed using the following Lemma:
\vspace{3mm}
\begin{lemma}
\label{lema3}
Let us consider the communication between two legitimate peers A and B in the presence of an external eavesdropper E. Let  $\gamma_b$ and $\gamma_e$ be the instantaneous SNRs at B and E, respectively, and $\bar\gamma_b$ and $\bar\gamma_e$ the average SNRs at B and E, respectively. If ${\gamma_e\sim \mathcal{S}_{\kappa\mu m}(\gamma_e;\kappa_e,\mu_e,m_e)}$ with $\{\mu,m\}\in\mathbb{Z}^+$, then for any finite rate $R_S\geq0$ the OPSC can be expressed in terms of the IMGF of $\gamma_b$ as
\begin{align}
\label{eqProb}
\mathcal{P}_{R_S}&=\sum_{i=0}^{M}C_i e^{\alpha\beta_i}\sum_{r=0}^{m_i-1}\sum_{k=0}^{r}   \tfrac{(-\alpha\beta_i)^{r}}{k!(r-k)!}  \tfrac{\partial^k \M_{\gamma_b}^{u}\left( s,z\right)}{\partial s^k}\bigg{|}_{\stackrel{s=-\beta_i}{z=\alpha}}+\M_{\gamma_b}^{l}(0,2^{R_S}-1)\nonumber\\
&=\sum_{i=0}^{M}C_i e^{\alpha\beta_i}\sum_{r=0}^{m_i-1}\sum_{k=0}^{r}   \tfrac{(-\alpha\beta_i)^{r}}{k!(r-k)!}  \tfrac{\partial^k \M_{\gamma_b}^{u}\left( s,z\right)}{\partial s^k}\bigg{|}_{\stackrel{s=-\beta_i}{z=\alpha}}+F_{\gamma_b}\left(\alpha\right).
\end{align}
where $F_{\gamma_b}(\cdot)$ is the CDF of $\gamma_b$, $\alpha=2^{R_S}-1$, $\beta_i=1/(2^{R_S}\Omega_i)$, and the parameters $M$, $m_i$, $\Omega_i$ and $C_i$ are related to $\kappa$, $\mu$, $m$ and $\bar\gamma_e$ as described in Table \ref{table02} in Appendix \ref{app2}.
\end{lemma}
\vspace{3mm}
\begin{proof}
See Appendix \ref{app1}.
\end{proof}
\vspace{3mm}

Expression (\ref{eqfinal01}) yields the OPSC in any scenario on which the eavesdropper's fading channel can be modeled with the $\kappa$-$\mu$ shadowed distribution, for \textit{any arbitrary} choice of the fading distribution for the legitimate channel. Thence, we can use the results in Table \ref{table01} combined with (\ref{obvious}) to derive analytical expressions for the secrecy performance in those scenarios on which the fading at the desired link can be modeled by the general $\kappa$-$\mu$ shadowed distribution, or any of the particular cases included therein. Note that there is no need to restrict the parameters $\mu$ and $m$ of the legitimate channel to take integer values, as the IMGF derived in Table \ref{table01} holds for any $\{\mu,m\}\in\mathbb{R}$.

In some scenarios, Eve may only have access to signals arriving from NLOS paths \cite{Maleki2014,Classen2015}. Under this premise, we can assume that the eavesdropper link can be modeled by the Rayleigh distribution. This leads $\gamma_e$ to be exponentially distributed with average SNR $\bar\gamma_e$, i.e. $\gamma_e\sim\text{Exp}{(\bar\gamma_e)}$. Thus, the following corollary arises as a special case of Lemma 2.

\vspace{3mm}
\begin{corollary}
\label{coroR1}
Let us consider the communication between two legitimate peers A and B in the presence of an external eavesdropper E. Let  $\gamma_b$ and $\gamma_e$ be the instantaneous SNRs at B and E, respectively, and $\bar\gamma_b$ and $\bar\gamma_e$ the average SNRs at B and E, respectively. If ${\gamma_e\sim\text{Exp}{(\bar\gamma_e)}}$, then for any finite rate $R_S\geq0$ the OPSC can be expressed in terms of the IMGF of $\gamma_b$ as
\begin{align}
\label{eqProb2}
\mathcal{P}_{R_S}=&\M_{\gamma_b}^{l}(0,2^{R_S}-1)+e^{\frac{2^{R_S}-1}{2^{R_S}\bar\gamma_e}}\M_{\gamma_b}^{u}\left(-\frac{1}{2^{R_S}\bar\gamma_e},2^{R_S}-1\right)\nonumber\\
=&F_{\gamma_b}\left(2^{R_S}-1\right)+e^{\frac{2^{R_S}-1}{2^{R_S}\bar\gamma_e}}\M_{\gamma_b}^{u}\left(-\frac{1}{2^{R_S}\bar\gamma_e},2^{R_S}-1\right),
\end{align}
where $F_{\gamma_b}(\cdot)$ is the CDF of $\gamma_b$.
\end{corollary}
\vspace{3mm}
\begin{proof}
Following the same derivation in Appendix \ref{app1} and setting $\kappa=0$ and $\mu=1$ yields the desired result.
\end{proof}
\vspace{3mm}

Note that these results provide a systematic way to derive the OPSC for any arbitrary fading distribution in the legitimate link, provided that the IMGF of the SNR at Bob is known. We must also note the OPSC for $\bar\gamma_b\gg\bar\gamma_e$ does not depend on the distribution of $\gamma_e$, but only on the distribution of $\gamma_b$ and the average SNR at Eve $\bar\gamma_e$ \cite{Romero2014}. Thus, the consideration of Rayleigh fading for the eavesdropper link as performed in Corollary 5 has a negligible effect in practice, while simplifying the analysis.

If now assuming Rayleigh fading also for the legitimate channel, the OPSC expression in (\ref{eqProb2}) reduces to the one originally calculated in \cite{Barros2006}. This can be checked by setting $\kappa=0$ and $\mu=1$ in the IMGF of the $\kappa$-$\mu$ distribution in the third entry of Table \ref{table01}. Using the equivalence $Q_1(0,\sqrt{2x})=e^{-x}$ given in \cite[eq. 4.45]{AlouiniBook}, and after some manipulations we obtain
\begin{align}
\label{eqBarros}
\mathcal{P}_{R_S}=1-e^{-\frac{2^{R_S}-1}{\bar\gamma_b}}\frac{\bar\gamma_b}{\bar\gamma_b+2^{R_S}\bar\gamma_e}.
\end{align}

}
\subsection{Extension to the multi-antenna scenario}
\label{app2}

In the previous analysis, we have explicitly assumed that all the agents in the system are equipped with single-antenna devices. We here show that the extension for the case on which Eve and Bob are multi-antenna devices can be straightforwardly carried out. 

First, note that in the derivation carried out in Appendix \ref{app1} in order to prove Lemma 2, there is neither any restriction related to the number of antennas used by Bob, nor related to the combining strategy carried out. Hence, Lemma 2 can be applied as is for a multi-antenna configuration at Bob. In this case, the only requirement to derive the OPSC is to determine the IMGF of the SNR after combining. For instance, if maximal ratio combining (MRC) is used by Bob we have that $\gamma_b=\sum_{i=1}^{N_B}{\gamma_b}_i$ and
\begin{align}
\M_{\gamma_b}^{u}\left(s,z\right)=\prod_{i=1}^{N_B}\M_{{\gamma_b}_i}^{u}\left(s_i,z_i\right),
\end{align}
where $N_B$ is the number of receive antennas at Bob, ${{\gamma_b}_i}$ denote the per-branch instantaneous SNRs, and independence between receive branches has been considered.

When assuming multiple antennas at Eve with i.i.d. branches under {$\kappa$-$\mu$ shadowed fading, the extension is also straightforward when considering that Eve performs MRC reception in order to maximize the receive SNR; since the sum of $N_E$ i.i.d. $\kappa$-$\mu$ shadowed random variables is also $\kappa$-$\mu$ shadowed-distributed with $\mu_{eq}=N_E\cdot{\mu}$, $m_{eq}=N_E\cdot{m}$ and $\gamma_e=N_E\cdot{\gamma_e}$, the OPSC in the multiantenna scenario can be expressed as in Lemma 2.}

{
\section{Other Applications}
\label{analysis2}

\subsection{Outage probability analysis with interference and background noise}

The performance characterization of wireless communication systems in the presence of interference is a very important problem in communication theory, ever since the advent of digital cellular systems, whose performance is known to be limited by the interference received from nearby cells. Let us denote as $\gamma_d$ the instantaneous SNR at the intended receiver, and let us denote as $\gamma_i$ the aggregate instantaneous interference-to-noise ratio corresponding to the set of interfering signals affecting the receiver.

The outage probability in this scenario, defined as the probability that the signal-to-noise plus interference ratio is below a given threshold $\gamma_{\text{th}}$,  can be calculated as 
\begin{align}
\label{OPNI}
\text{OP}_{\text{NI}}=\Pr\left\{ \gamma_d  \leqslant \gamma_{\text{th}}\left(\gamma_i+1\right) \right\},
\end{align}

As pointed out in \cite{Gomez2015}, expressions (\ref{OPNI}) and (\ref{eq_analysis_1}) are formally equivalent, up to some scaling of the random variables and setting $\gamma_{\text{th}}=2^{R_S}-1$. This means that both $\text{OP}_{\text{NI}}$ and $\mathcal{P}_{R_S}$ will have the same functional form for the same choice of distributions.

The derivation of the outage probability of communication systems in the presence of interference and background noise in a $\kappa$-$\mu$ shadowed/$\kappa$-$\mu$ shadowed scenario, which is an open problem in the literature, is directly obtained by using the IMGF given in the first entry of Table I, under the same conditions assumed in Section \ref{analysis}. Recent results in the literature arise as special cases \cite{Ermolova2014}.
%
%

\subsection{Channel capacity with side information at the transmitter and the receiver}

Let us now consider the scenario on which a transmitter, subject to an average transmit power constraint, communicates with a receiver through a fading channel. Assuming perfect channel knowledge at both the transmitter and receiver sides, the transmitter can optimally adapt its power and rate. The Shannon capacity in this scenario is known to be given by the following expression \cite{Goldsmith1997},

\begin{equation}
\label{cap}
C= \int_{0}^{\infty}\log\left(\frac{\gamma}{\gamma_0}\right)f_{\gamma}(\gamma)d\gamma,
\end{equation}
where a normalized bandwidth $B=1$ was assumed for simplicity. In (\ref{cap}), $\gamma$ is the instantaneous SNR and $\gamma_0$ is a cut-off SNR determined by the average power constraint. An alternative expression for this capacity was proposed in \cite{DiRenzo2010} in terms of the $E_i$-transform, which makes use of the exponential integral function $E_i(\cdot)$ as integration kernel, yielding

\begin{equation}
\label{cap2}
C= \frac{1}{\log(2)} \int_{0}^{\infty}E_i(-x)\exp(x)\Psi(x,\gamma_0)dx,
\end{equation}
where the ancillary function $\Psi(x,\gamma_0)$ is defined as
\begin{equation}
\label{cap3}
\Psi(x,\gamma_0)\triangleq M_{\gamma}^{u}\left(\frac{x}{\gamma_0},x\right) + \frac{1}{\gamma_0} \frac{\partial \M_{\gamma}^{u}\left( s,z\right)}{\partial s}\bigg{|}_{\stackrel{s=\frac{x}{\gamma_0}}{z=x}}.
\end{equation}

Thus, equation (\ref{cap2}) provides an alternative way of computing the capacity in this scenario for an arbitrary distribution of the SNR, in terms of the IMGF and its first derivative. Using the expressions for the IMGF derived in Table I, capacity results are obtained for the $\kappa$-$\mu$ shadowed fading channels, and all the special cases included therein. These results are also new in the literature.

\subsection{Average bit-error rate with adaptive modulation}

Adaptive modulation makes use of channel knowledge at the transmitter side in order to optimally design system parameters such as constellation size, transmit power, coding rates and schemes, and many others \cite{Gold05}. One extended alternative is the design of the constellation size and power in order to maximize the average throughput, for a certain instantaneous bit-error rate (BER) constraint. In this scenario, the average BER $\bar P_b$ of adaptive modulation with $M$-QAM is well approximated using \cite[eq. 9.7]{Gold05} and \mbox{\cite[eq. 9.72]{Gold05}}, as
\begin{equation}
\label{ABER}
\bar P_b  \approx \frac{{\sum\limits_{j = 1}^{N - 1} {k_j \int_{\gamma _{j - 1} }^{\gamma _j } {0.2\exp \left( { - 1.5\frac{{\gamma }}
{{2^{k_j }  - 1}}} \right)f_{\gamma}\left( \gamma  \right)d\gamma } } }}
{{\sum\limits_{j = 1}^{N - 1} {k_j \int_{\gamma _{j - 1} }^{\gamma _j } {f_{\gamma}\left( \gamma  \right)d\gamma } } }},
\end{equation}
where $\gamma$ represents the instantaneous SNR, $\bar\gamma$ is the average SNR, $N$ is the number of fading regions, $\{\gamma_j\}$ are the SNR switching thresholds and $k_j$ is the number of bits per complex symbol employed when $\gamma_{j-1}\leq\gamma<\gamma_j$. Note that the denominator in (\ref{ABER}) represents the exact average spectral efficiency and, for convenience, $\gamma_{N-1}\triangleq \infty$. Using the IMGF of $\gamma$, the following closed-form expression is obtained as
\begin{align}
\label{ABER2}
\bar P_b  &\approx \frac{{0.2\sum\limits_{j = 1}^{N - 1} {k_j \left\{ {\M_{\gamma}^{l}(- \frac{{1.5}}
{{2^{k_j }  - 1}},\gamma _{j - 1}) - \M_{\gamma}^{l}(- \frac{{1.5}}
{{2^{k_j }  - 1}},\gamma _{j})} \right\}} }}
{{\sum\limits_{j = 1}^{N - 1} {k_j \left\{ {\M_{\gamma}^{l}(0,\gamma _{j - 1}) - \M_{\gamma}^{l}(0,\gamma _{j})} \right\}} }}\nonumber\\
&\approx \frac{{0.2\sum\limits_{j = 1}^{N - 1} {k_j \left\{ {\M_{\gamma}^{l}(- \frac{{1.5}}
{{2^{k_j }  - 1}},\gamma _{j - 1}) - \M_{\gamma}^{l}(- \frac{{1.5}}
{{2^{k_j }  - 1}},\gamma _{j})} \right\}} }}
{{\sum\limits_{j = 1}^{N - 1} {k_j \left\{ {F_{\gamma}(\gamma _{j - 1}) - F_{\gamma}(\gamma _{j})} \right\}} }}.
\end{align}

Therefore, the average BER in this scenario for \textit{any} fading distribution can be easily obtained by evaluating a finite number of terms involving the IMGF. More specifically, closed-form results new in the literature can be obtained for the case of the $\kappa$-$\mu$ shadowed fading distribution and special cases.
}

\section{Numerical Results}
\label{numerical}

{In this section we provide numerical results for some of the practical scenarios previously analyzed. Specifically, we focus on} the outage probability of the secrecy capacity studied in Section \ref{analysis} under different fading scenarios. We assume that Bob and Eve are only equipped with one antenna, {and for the eavesdropper's channel we set $\kappa=0$ and $\mu=1$}. The effect of system parameters on the $\epsilon$-outage secrecy capacity is also investigated. All the results shown here have been analytically obtained by the direct evaluation of the expressions developed in this paper: Additionally, Monte Carlo simulations have been performed to validate the derived expressions, and are also presented in all figures, showing an excellent agreement with the analytical results. Details on how to compute the confluent bivariate function $\Phi_2$ are given in \cite[App. E]{Paris2014}.

In Figs. \ref{fig1}-\ref{fig5}, the OPSC is represented considering different fading models as a function of the average SNR at Bob $\bar\gamma_b$, for different sets of values of the fading parameters. We assume in these figures that the normalized rate threshold value used to declare an outage is $R_S = 0.1$, and an average SNR at Eve $\bar\gamma_e$ = 15 dB. 

Figs. \ref{fig1} and \ref{fig2}, show results for the $\kappa$-$\mu$ shadowed fading considering, respectively, small ($\kappa=1.5$) and large ($\kappa=10$) LOS components in the received wave clusters for different values of the $\mu$ parameter and also considering light ($m=12$) or heavy ($m=0.5$) shadowing for the LOS components. As expected, as the fading parameter $\mu$ increases, the diversity gain increases too, resulting in a higher slope of the curves in the high SNR regime, and with diminishing returns as $\mu$ increases. Note that $\mu$ represents the number of received wave clusters when it takes an integer value. It can also be observed that the performance is always better when the LOS components are lightly shadowed, and this improvement is much more noticeable for large LOS components.

\begin{figure}[t]
\includegraphics[width=.99\columnwidth]{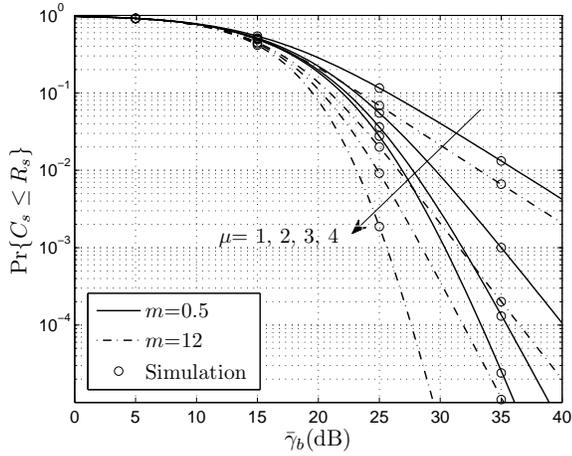}
\caption{Outage probability of secrecy capacity under $\kappa$-$\mu$ shadowed fading as a function of $\bar\gamma_b$, for different values of $m$ and $\mu$. Parameter values: $\kappa=1.5$, $\bar\gamma_e=15$ dB and $R_S=0.1$.}
\label{fig1}
\end{figure}

\begin{figure}[t]
\includegraphics[width=.99\columnwidth]{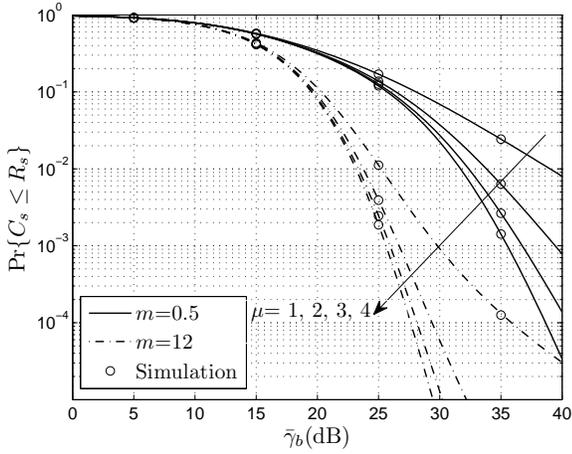}
\caption{Outage probability of secrecy capacity under $\kappa$-$\mu$ shadowed fading as a function of $\bar\gamma_b$, for different values of $m$ and $\mu$. Parameter values: $\kappa=10$, $\bar\gamma_e=15$ dB and $R_S=0.1$.}
\label{fig2}
\end{figure}

The impact of shadowed LOS components on performance can be observed in Fig. \ref{fig3}, where the outage probability of the secrecy capacity under Rician shadowed fading is presented for different values of the $m$ and $K$ parameters. It can be observed that it is more beneficial for the performance to have small LOS components ($K=1.5$) if they are affected by heavy shadowing. Conversely, if the shadowing is mild, large LOS components always yield a lower outage probability. This appreciation can be confirmed by observing the results in Fig. \ref{fig4}, which depicts the outage probability under $\kappa$-$\mu$ fading, i.e., when the LOS components do not experience any shadowing.

\begin{figure}[t]
\includegraphics[width=.99\columnwidth]{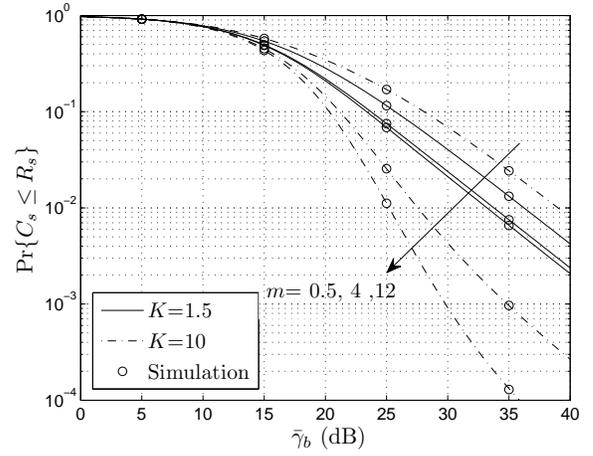}
\caption{Outage probability of secrecy capacity under Rician shadowed fading as a function of $\bar\gamma_b$, for different values of $K$ and $m$. Parameter values: $\bar\gamma_e=15$ dB and $R_S=0.1$.}
\label{fig3}
\end{figure}

\begin{figure}[t]
\includegraphics[width=.99\columnwidth]{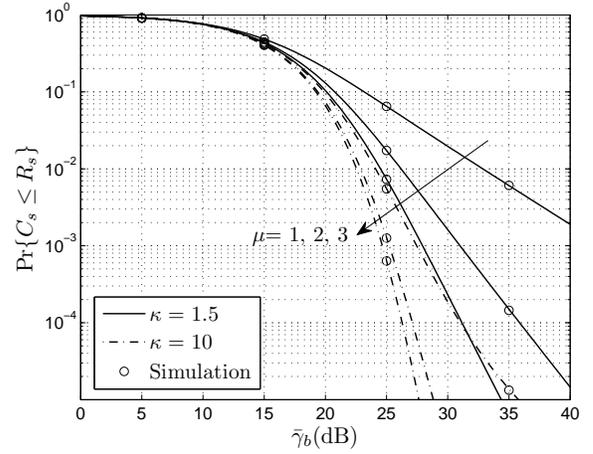}
\caption{Outage probability of secrecy capacity under $\kappa$-$\mu$ fading as a function of $\bar\gamma_b$, for different values of $\kappa$ and $\mu$. Parameter values: $\bar\gamma_e=15$ dB and $R_S=0.1$.}
\label{fig4}
\end{figure}

Fig. \ref{fig5} presents results for the $\eta$-$\mu$ fading, i.e.. in a NLOS scenario on which the in-phase and quadrature components of the scattered waves are not necessarily equally distributed. We consider format 1 of this distribution, for which $\eta$ represents the scattered-wave power ratio between the in-phase and quadrature components of each cluster of multipath, and the number of multipath clusters, when $\mu$ is a semi-integer, is represented by $2\mu$ . It can be observed that, when the in-phase and quadrature components are highly imbalanced ($\eta=0.04$), the performance is poorer. On the other hand, increasing the number of multipath clusters have always a beneficial impact on performance, as the instantaneous  received signal power is smoothed.

\begin{figure}[t]
\includegraphics[width=.99\columnwidth]{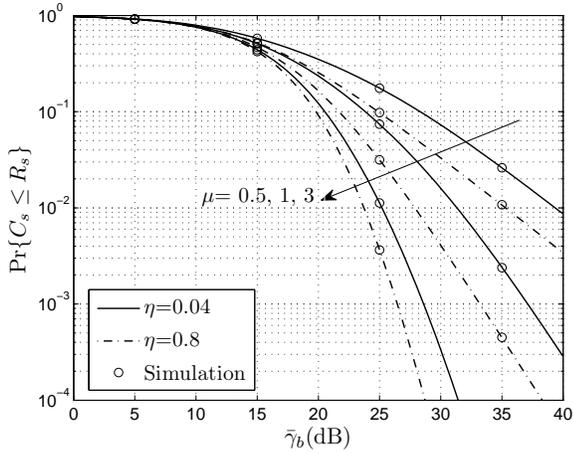}
\caption{Outage probability of secrecy capacity under $\eta$-$\mu$ fading as a function of $\bar\gamma_b$, for different values of $\eta$ and $\mu$. Parameter values: $\bar\gamma_e=15$ dB and $R_S=0.1$.}
\label{fig5}
\end{figure}

Fig. \ref{fig6} shows the normalized $\epsilon$-outage secrecy capacity $C_{\epsilon}$ under $\kappa$-$\mu$ shadowed fading as a function of $\bar\gamma_b$, where the capacity normalization has been computed with respect to the capacity of an AWGN channel with SNR equal to $\bar\gamma_b$. We have assumed $m=2$ and an average SNR at Eve of $\bar\gamma_e=-10$ dB. The corresponding results for $\eta$-$\mu$ fading are presented in Fig. \ref{fig7}, also for $\bar\gamma_e=-10$ dB. In both figures it can be observed that when the outage probability is set to a high value ($\epsilon=0.8$), better channel conditions (i.e. $\mu=6$, $\kappa=10$ for $\kappa$-$\mu$ shadowed fading and $\mu=4$, $\eta=0.9$ for $\eta$-$\mu$ fading) yield a lower $\epsilon$-outage capacity. Conversely, better channel conditions results in a higher capacity for lower values of $\epsilon$. 
Further insight can be obtained from Fig. \ref{fig8}, where it is shown the normalized $\epsilon$-outage secrecy capacity $C_{\epsilon}$ under $\kappa$-$\mu$ fading as a function of the outage probability $\epsilon$ and for different average SNRs at Eve, assuming an average SNR of 10 dB at the desired link. We observe that higher outage probability $\epsilon$ leads to higher $C_{\epsilon}$, having an important influence the average SNR of the eavesdropper's channel. It can also be observed that, as the channel conditions improve, the normalized $\epsilon$-outage secrecy capacity tends to one for all values of the outage probability for low values of $\bar\gamma_e$.

\begin{figure}[t]
\includegraphics[width=.99\columnwidth]{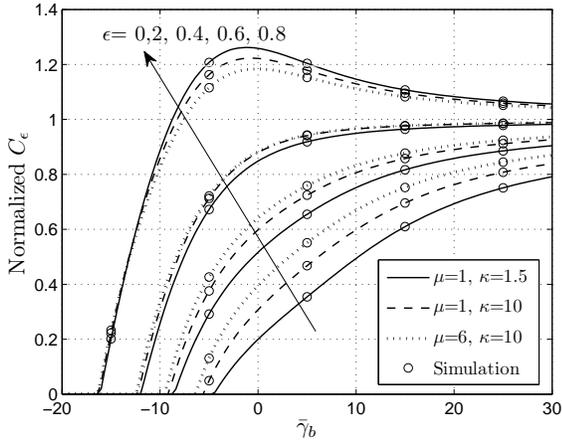}
\caption{Normalized $\epsilon$-outage secrecy capacity $C_{\epsilon}$ under $\kappa$-$\mu$ shadowed fading as a function of $\bar\gamma_b$. Parameter values; $m=2$, $\bar\gamma_e=-10$ dB.}
\label{fig6}
\end{figure}

\begin{figure}[t]
\includegraphics[width=.99\columnwidth]{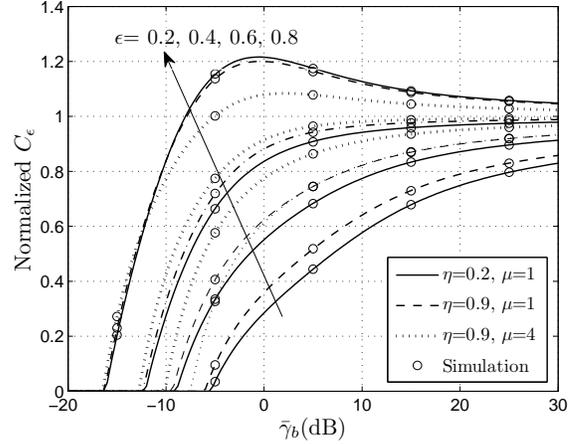}
\caption{Normalized $\epsilon$-outage secrecy capacity $C_{\epsilon}$ under $\eta$-$\mu$ fading as a function of $\bar\gamma_b$. Parameter value $\bar\gamma_e=-10$ dB.}
\label{fig7}
\end{figure}

\begin{figure}[ht]
\includegraphics[width=\columnwidth]{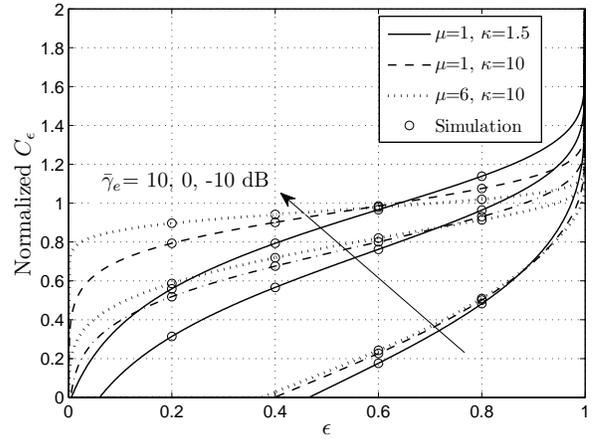}
\caption{Normalized $\epsilon$-outage secrecy capacity $C_{\epsilon}$ under $\kappa$-$\mu$ fading as a function of $\epsilon$. Parameter value $\gamma_b=10$ dB.}
\label{fig8}
\end{figure}

Note that wireless communication over fading channels does not require necessarily the average SNR of the channel between Alice and Bob to be greater than the average SNR of the channel between Alice and Eve, since there is certain probability that the instantaneous SNR of the main channel being higher than the instantaneous SNR of the eavesdropper's channel ($\gamma_b>\gamma_e$) even when $\bar\gamma_b<\bar\gamma_e$. In fact, there is a trade off between the outage probability of the secrecy capacity ${ \Pr \left\{C_S\leq R_S\right\}}$ and the $\epsilon$-outage secrecy capacity $C_{\epsilon}$, where a higher $C_{\epsilon}$ corresponds to a higher outage probability $\epsilon$, and viceversa. For that reason, results in Figs. \ref{fig6}, \ref{fig7}  and \ref{fig8} show that the normalized outage secrecy capacity may have non-zero values even when $\bar\gamma_b \leq \bar\gamma_e$ (for high values of $\epsilon$).

\section{Conclusions}
\label{conclusions}

A fundamental connection between the incomplete MGF of a positive random variable and its complete MGF has been presented. The main takeaway is that the IMGF is expected to have a similar functional form as the CDF, and hence its evaluation should not require any additional complexity. Using this novel connection, closed-form expressions for the IMGF of the $\kappa$-$\mu$ shadowed distribution (and all the special cases included therein) have been derived for the first time in the literature. This has enabled us to introduce a new framework for the analysis of the physical layer security in scenarios on which the desired link is affected by any arbitrary fading distribution, {and the eavesdropper's link undergoes $\kappa$-$\mu$ shadowed fading. 

We hope that the results in this paper may facilitate the performance evaluation for more practical setups related to physical layer security, such as those using artificial noise transmission or the collaboration of a friendly jammer, and, in general for all the situations and scenarios on which the IMGF makes appearances.}

\section*{Acknowledgment}
This work has been funded by the Consejer\'ia de Econom\'ia, Innovaci\'on, Ciencia y Empleo of the Junta de Andaluc\'ia, the Spanish Government
and the European Fund for Regional Development FEDER (projects P2011-TIC-7109, P2011-TIC-8238, TEC2013-42711-R, TEC2013-44442-P and TEC2014-57901-R).

\appendices
\section{Proof of Lemma 2}
\label{app1}


From the definition of OPSC in (\ref{eq_analysis_1}), the probability of achieving a successful secure communication is given by
\begin{align}
\label{eq:sec30}
& \Pr\left\{C_S>R_S\right\}= \Pr \left\{\gamma_e<\frac{1}{2^{R_S}}(1+\gamma_b)-1\right\} .
\end{align}
Note that $\gamma_e$ only takes non-negative values; hence, for this condition to occur in (\ref{eq:sec30}), then the inequality ${\gamma_b \geq 2^{R_S}-1}$ must be satisfied, i.e.
\begin{align}
\label{eq:sec31}
\Pr \left\{\gamma_e<\frac{1}{2^{R_S}}(1+\gamma_b)-1 | \gamma_b < 2^{R_S}-1\right\} = 0.
\end{align}
Therefore we can write
\begin{align}
\label{eq:sec34}
\Pr\left\{C_S>R_S\right\}&= \int_{2^{R_S}-1}^\infty f_{\gamma_b}(x) \left(\int_0^{\frac{1}{2^{R_S}}(1+x)-1} f_{\gamma_e}(y) dy\right) dx  \nonumber \\&=
\int_{2^{R_S}-1}^\infty f_{\gamma_b}(x) F_{\gamma_e}\left(\frac{1}{2^{R_S}}(1+x)-1\right) dx,
\end{align}
where $F_{X}(\cdot)$ and  $f_{X}(\cdot)$ are the CDF and PDF of the random variable $X$, respectively. 

{
Let us first assume that the fading experienced by the eavesdropper's is modeled by th $\kappa$-$\mu$ shadowed distribution \cite{Paris2014}. For integer values of $\mu$ and $m$, the CDF of the $\kappa$-$\mu$ shadowed distribution can be expressed as a mixture of gamma distributions as described in Appendix \ref{app2}, as follows:

\begin{equation}
\label{eqmulti}
F_{\gamma_e}(\gamma_e)=1-\sum_{i=0}^{M}C_i e^{-\frac{\gamma}{\Omega_i}}\sum_{r=0}^{m_i-1}\frac{1}{r!}\left(\frac{\gamma}{\Omega_i}\right)^{r},
\end{equation}

Plugging (\ref{eqmulti}) in (\ref{eq:sec34}) yields
\begin{align}
\label{eqapp2b1}
&\Pr\left\{C_S>R_S\right\}= 1- F_{\gamma_b}(2^{R_S}-1)
-\sum_{i=0}^{M}C_i  \exp{\left(-\frac{1-2^{R_S}}{2^{R_S}\bar \Omega_i}\right)} \nonumber
\\& \times \underbrace{\int_{2^{R_S}-1}^\infty f_{\gamma_b}(x) e^{ -\frac{ x}{2^{R_S}  \bar \Omega_i}} \sum_{r=0}^{m_i-1}   \frac{1}{r!}\left(\frac{1+x-2^{R_S}}{2^{R_S}\Omega_i}\right)^r dx}_{\mathcal{I}},
\end{align}
 
where the parameters $M$, $m_i$, $C_i$ and $\Omega_i$ are defined in Table \ref{table02} in the Appendix \ref{app2} in terms of the parameters $\kappa$, $\mu$, $m$ and $\gamma_e$ of the eavesdropper's fading distribution. Using the binomial expansion, and defining $\alpha=2^{R_S}-1$, $\beta_i=1/(2^{R_S}\Omega_i)$, the integral term $\mathcal{I}$ can be reexpressed as
\begin{align}
\label{eqapp2b2}
\mathcal{I} &=\sum_{r=0}^{m_i-1}\sum_{k=0}^{r}   \frac{(-\alpha)^{r-k}}{k!(r-k)!} \beta_i^{r} \int_{\alpha}^\infty x^k f_{\gamma_b}(x) e^{ -\beta_i x} dx,\nonumber\\
&=\sum_{i=0}^{m_i-1}\sum_{k=0}^{r}   \frac{\alpha^{r-k}}{k!(r-k)!} \beta_i^{r} \frac{\partial^k \M_{\gamma_b}^{u}\left( s,z\right)}{\partial s^k}\bigg{|}_{\stackrel{s=-\beta_i}{z=\alpha}}
\end{align}
where the general derivative property in the transform domain was used, in order to identify the $k^{th}$ derivative of the IMGF. Finally, using (\ref{eqapp2b2}) in (\ref{eqapp2b1}) and (\ref{eq_analysis_1}) yields
\begin{align}
\label{eqfinal01}
\mathcal{P}_{R_S}= \sum_{i=0}^{M}C_i e^{\alpha\beta_i}\sum_{r=0}^{m_i-1}\sum_{k=0}^{r}   \tfrac{(-\alpha\beta_i)^{r}}{k!(r-k)!}  \frac{\partial^k \M_{\gamma_b}^{u}\left( s,z\right)}{\partial s^k}\bigg{|}_{\stackrel{s=-\beta_i}{z=\alpha}}+F_{\gamma_b}\left(\alpha\right).
\end{align}

This completes the proof. }

{
\section{CDF of the $\kappa$-$\mu$ shadowed distribution for integer fading parameters}
\label{app2}


The CDF of the $\kappa$-$\mu$ shadowed fading model was originally given in \cite{Paris2014} as

\begin{align}
F_{\gamma}(\gamma)=&\frac{{\mu ^{\mu-1}  m^m \left( {1 + \kappa } \right)^\mu  }}
{{\Gamma(\mu)\bar \gamma ^\mu  \left( {\mu \kappa  + m} \right)^m }}z^\mu   \times \nonumber\\
 & \Phi _2 \left( {\mu  - m,m;1 + \mu ; { - \tfrac{{\mu \left( {1 + \kappa } \right)}}
{{\bar \gamma }}} \gamma, { - \tfrac{{\mu \left( {1 + \kappa } \right)}}
{{\bar \gamma }}\tfrac{m}
{{\mu \kappa  + m}}} \gamma} \right).
\end{align}

If the fading parameters $\mu$ and $m$ take integer values, this CDF can be expressed as a finite mixture of squared Nakagami distributions, i.e. as a finite sum of exponentials and powers \cite{Lopez2016}. Manipulating the expressions in \cite{Lopez2016}, we can compactly express the CDF as

\begin{equation}
F_{\gamma}(\gamma)=1-\sum_{i=0}^{M}C_i e^{-\frac{\gamma}{\Omega_i}}\sum_{r=0}^{m_i-1}\frac{1}{r!}\left(\frac{\gamma}{\Omega_i}\right)^{r},
\end{equation}

where the parameters $m_i$, $M$ and $\Omega_i$ are expressed in Table \ref{table02} in terms of the parameters of the $\kappa$-$\mu$ shadowed distribution, namely $\kappa$, $\mu$, $m$ and $\bar\gamma$.
}
\begin{table*}[h]
  \renewcommand{\arraystretch}{3}
\centering
{
\caption{Parameter values for the $\kappa$-$\mu$ shadowed distribution with integer $\mu$ and $m$,}
\label{table02}
\begin{tabular}{|c|c|}
\hline\hline
Case $\mu>m$ & Case $\mu \leq m$ \\ \hline\hline 
$M=\mu$ & $M=m-\mu+1$  \\ \hline
 $C_i=\begin{cases} 
      0 & i=0 \\
       \left( { - 1} \right)^m \binom{m+i-2}{i-1}\times \left[ {\frac{m}
{{\mu \kappa  + m}}} \right]^{ m} \left[ {\frac{{\mu \kappa }}
{{\mu \kappa  + m}}} \right]^{ - m - j + 1}  & 0<i\leq \mu-m \\
      \left( { - 1} \right)^{i-\mu+m - 1} \binom{i-2}{i-\mu+m-1} \times \left[ {\frac{m}
{{\mu \kappa  + m}}} \right]^{i-\mu+m - 1} \left[ {\frac{{\mu \kappa }}
{{\mu \kappa  + m}}} \right]^{-i + 1}  & \mu-m < i \leq \mu 
   \end{cases}$
 & $C_i=\binom{m-\mu}{j}\left[ {\frac{m}
{{\mu \kappa  + m}}} \right]^j \left[ {\frac{{\mu \kappa }}
{{\mu \kappa  + m}}} \right]^{m - \mu  - j}$  \\ \hline
  $\Omega_i=\begin{cases}  \frac{{\bar \gamma }}
{{\mu \left( {1 + \kappa } \right)}}, & 0\leq i\leq \mu-m \\
   \frac{{\mu \kappa  + m}}
{m}\frac{{\bar \gamma }}
{{\mu \left( {1 + \kappa } \right)}} & \mu-m < i \leq \mu 
   \end{cases}$ & $\Omega_i=\frac{{\mu \kappa  + m}}
{m}\frac{{\bar \gamma }}
{{\mu \left( {1 + \kappa } \right)}}$  \\ \hline
\ $m_i=\begin{cases} \mu-m-i+1, & 0\leq i\leq \mu-m \\
   \mu-i+1 & \mu-m < i \leq \mu 
   \end{cases}$ & $m_i=m-i$  \\ \hline
\end{tabular}
}
\end{table*}

\bibliographystyle{IEEEtran}
\bibliography{bibfile}

\begin{thebibliography}{10}
\providecommand{\url}[1]{#1}
\csname url@samestyle\endcsname
\providecommand{\newblock}{\relax}
\providecommand{\bibinfo}[2]{#2}
\providecommand{\BIBentrySTDinterwordspacing}{\spaceskip=0pt\relax}
\providecommand{\BIBentryALTinterwordstretchfactor}{4}
\providecommand{\BIBentryALTinterwordspacing}{\spaceskip=\fontdimen2\font plus
\BIBentryALTinterwordstretchfactor\fontdimen3\font minus
  \fontdimen4\font\relax}
\providecommand{\BIBforeignlanguage}[2]{{%
\expandafter\ifx\csname l@#1\endcsname\relax
\typeout{** WARNING: IEEEtran.bst: No hyphenation pattern has been}%
\typeout{** loaded for the language `#1'. Using the pattern for}%
\typeout{** the default language instead.}%
\else
\language=\csname l@#1\endcsname
\fi
#2}}
\providecommand{\BIBdecl}{\relax}
\BIBdecl

\bibitem{Simon1998}
M.~K. Simon and M.-S. Alouini, ``A unified approach to the performance analysis
  of digital communication over generalized fading channels,''
  \emph{Proceedings of the IEEE}, vol.~86, no.~9, pp. 1860--1877, Sep 1998.

\bibitem{DiRenzo2010}
M.~Di~Renzo, F.~Graziosi, and F.~Santucci, ``{Channel Capacity Over Generalized
  Fading Channels: A Novel MGF-Based Approach for Performance Analysis and
  Design of Wireless Communication Systems},'' \emph{IEEE Trans. Veh.
  Technol.}, vol.~59, no.~1, pp. 127--149, Jan 2010.

\bibitem{Yilmaz2012}
F.~Yilmaz and M.-S. Alouini, ``{A Unified MGF-Based Capacity Analysis of
  Diversity Combiners over Generalized Fading Channels},'' \emph{IEEE Trans.
  Commun.}, vol.~60, no.~3, pp. 862--875, March 2012.

\bibitem{Rao2015}
M.~Rao, F.~J. Lopez-Martinez, M.~S. Alouini, and A.~Goldsmith, ``{MGF Approach
  to the Analysis of Generalized Two-Ray Fading Models},'' \emph{IEEE Trans.
  Wireless Commun.}, vol.~14, no.~5, pp. 2548--2561, May 2015.

\bibitem{AlouiniBook}
\BIBentryALTinterwordspacing
M.~K. Simon and M.-S. Alouini, \emph{{Digital communication over fading
  channels}}.\hskip 1em plus 0.5em minus 0.4em\relax {Wiley-IEEE Press}, 2005.
  [Online]. Available: \url{http://www.worldcat.org/isbn/0471649538}
\BIBentrySTDinterwordspacing

\bibitem{Ermolova2008}
N.~Ermolova, ``{Moment Generating Functions of the Generalized $\eta$-$\mu$ and
  $\kappa$-$\mu$ Distributions and Their Applications to Performance
  Evaluations of Communication Systems},'' \emph{IEEE Commun. Lett.}, vol.~12,
  no.~7, pp. 502--504, July 2008.

\bibitem{Paris2014}
J.~F. Paris, ``{Statistical Characterization of $\kappa$-$\mu$ Shadowed
  Fading},'' \emph{IEEE Trans. Veh. Technol.}, vol.~63, no.~2, pp. 518--526,
  Feb 2014.

\bibitem{Nuttall2002}
A.~H. Nuttall, ``{Joint probability density function of selected order
  statistics and the sum of the remaining random variables},'' DTIC Document,
  Tech. Rep., 2002.

\bibitem{Nam2010}
S.~S. Nam, M.~S. Alouini, and H.~C. Yang, ``{An MGF-Based Unified Framework to
  Determine the Joint Statistics of Partial Sums of Ordered Random
  Variables},'' \emph{IEEE Trans. Inf. Theory}, vol.~56, no.~11, pp.
  5655--5672, Nov 2010.

\bibitem{Yang2011}
H.-C. Yang and M.-S. Alouini, \emph{Order statistics in wireless
  communications: diversity, adaptation, and scheduling in MIMO and OFDM
  systems}.\hskip 1em plus 0.5em minus 0.4em\relax Cambridge University Press,
  2011.

\bibitem{Nam2014}
S.~S. Nam, H.~C. Yang, M.~S. Alouini, and D.~I. Kim, ``{An MGF-Based Unified
  Framework to Determine the Joint Statistics of Partial Sums of Ordered i.n.d.
  Random Variables},'' \emph{IEEE Trans. Signal Process.}, vol.~62, no.~16, pp.
  4270--4283, Aug 2014.

\bibitem{Xiaodi2006}
Z.~Xiaodi and N.~Beaulieu, ``{Performance Analysis of Generalized Selection
  Combining in Generalized Correlated Nakagami-m Fading},'' \emph{IEEE Trans.
  Commun.}, vol.~54, no.~11, pp. 2103--2112, Nov 2006.

\bibitem{Loskot2009}
P.~Loskot and N.~Beaulieu, ``A unified approach to computing error
  probabilities of diversity combining schemes over correlated fading
  channels,'' \emph{IEEE Trans. Commun.}, vol.~57, no.~7, pp. 2031--2041, July
  2009.

\bibitem{Annamalai2001}
A.~Annamalai, C.~Tellambura, and V.~K. Bhargava, ``Simple and accurate methods
  for outage analysis in cellular mobile radio systems-a unified approach,''
  \emph{IEEE Trans. Commun.}, vol.~49, no.~2, pp. 303--316, Feb 2001.

\bibitem{Morales2012}
D.~Morales-Jimenez, J.~F. Paris, and A.~Lozano, ``{Outage Probability Analysis
  for MRC in $\eta$-$\mu$ Fading Channels with Co-Channel Interference},''
  \emph{IEEE Commun. Lett.}, vol.~16, no.~5, pp. 674--677, May 2012.

\bibitem{Gaaloul2012}
F.~Gaaloul, R.~Radaydeh, H.-C. Yang, and M.-S. Aluoini, ``{Adaptive Scheduling
  with Postexamining User Selection Under Nonidentical Fading},'' \emph{IEEE
  Trans. Veh. Technol.}, vol.~61, no.~9, pp. 4175--4183, Nov 2012.

\bibitem{Yang2015}
J.~Yang, L.~Chen, X.~Lei, K.~Peppas, and T.~Duong, ``Dual-hop cognitive
  amplify-and-forward relaying networks over $\eta$-$\mu$ fading channels,''
  \emph{IEEE Trans. Veh. Technol.}, vol.~PP, no.~99, pp. 1--12, 2015.

\bibitem{JuanmaJavi15}
J.~M. Romero-Jerez, G.~Gomez, and F.~J. Lopez-Martinez, ``On the outage
  probability of secrecy capacity in arbitrarily-distributed fading channels,''
  in \emph{European Wireless 2015}, May 2015, pp. 1--6.

\bibitem{Cotton2015}
S.~L. Cotton, ``{Human Body Shadowing in Cellular Device-to-Device
  Communications: Channel Modeling Using the Shadowed $\kappa$ -$\mu$ Fading
  Model},'' \emph{IEEE J. Sel. Areas Commun.}, vol.~33, no.~1, pp. 111--119,
  Jan 2015.

\bibitem{Yacoub07}
M.~Yacoub, ``The $\kappa$-$\mu$ distribution and the $\eta$-$\mu$
  distribution,'' \emph{IEEE Antennas and Propagation Magazine}, vol.~49,
  no.~1, pp. 68--81, Feb 2007.

\bibitem{Laureano2015}
L.~Moreno-Pozas, F.~J. Lopez-Martinez, J.~F. Paris, and E.~Martos-Naya, ``The
  $\kappa$-$\mu$ shadowed fading model: Unifying the $\kappa$-$\mu $ and $\eta
  $-$\mu$ distributions,'' \emph{to appear in IEEE Trans. Veh. Technol.}, 2016.

\bibitem{Abdi2003}
A.~Abdi, W.~Lau, M.-S. Alouini, and M.~Kaveh, ``A new simple model for land
  mobile satellite channels: first- and second-order statistics,'' \emph{IEEE
  Trans. Wireless Commun.}, vol.~2, no.~3, pp. 519--528, May 2003.

\bibitem{Wyner1975}
A.~D. Wyner, ``The wire-tap channel,'' \emph{The Bell System Technical
  Journal}, vol.~54, no.~8, pp. 1355--1387, Oct 1975.

\bibitem{Cheong1978}
S.~Leung-Yan-Cheong and M.~Hellman, ``{The Gaussian wire-tap channel},''
  \emph{IEEE Trans. Inf. Theory}, vol.~24, no.~4, pp. 451--456, Jul 1978.

\bibitem{Barros2006}
J.~Barros and M.~R.~D. Rodrigues, ``Secrecy capacity of wireless channels,'' in
  \emph{IEEE International Symposium on Information Theory}, Jul. 2006, pp.
  356--360.

\bibitem{Bloch2008}
M.~Bloch, J.~Barros, M.~R.~D. Rodrigues, and S.~W. McLaughlin, ``Wireless
  information-theoretic security,'' \emph{IEEE Trans. Inf. Theory}, vol.~54,
  no.~6, pp. 2515--2534, Jun. 2008.

\bibitem{Sarkar09}
M.~Z.~I. Sarkar, T.~Ratnarajah, and M.~Sellathurai, ``Secrecy capacity of
  {N}akagami-m fading wireless channels in the presence of multiple
  eavesdroppers,'' in \emph{Signals, Systems and Computers, 2009 Conference
  Record of the Forty-Third Asilomar Conference on}.\hskip 1em plus 0.5em minus
  0.4em\relax IEEE, 2009, pp. 829--833.

\bibitem{Liu13}
X.~Liu, ``Probability of strictly positive secrecy capacity of the
  {R}ician-{R}ician fading channel,'' \emph{IEEE Wireless Communications
  Letters}, vol.~2, no.~1, pp. 50--53, 2013.

\bibitem{Liu13b}
------, ``Probability of strictly positive secrecy capacity of the {W}eibull
  fading channel,'' in \emph{IEEE Global Communications Conference (GLOBECOM),
  2013}, Dec 2013, pp. 659--664.

\bibitem{Wang14}
L.~Wang, N.~Yang, M.~Elkashlan, P.~L. Yeoh, and J.~Yuan, ``{Physical Layer
  Security of Maximal Ratio Combining in Two-Wave With Diffuse Power Fading
  Channels},'' \emph{IEEE Trans. Inf. Forensics Security}, vol.~9, no. 1-2, pp.
  247--258, 2014.

\bibitem{Romero2014}
J.~M. Romero-Jerez and F.~J. Lopez-Martinez, ``{A new framework for the
  performance analysis of wireless communications under Hoyt (Nakagami-q)
  fading},'' \emph{arXiv preprint arXiv:1403.0537}, 2014.

\bibitem{Bhargav2016}
N.~Bhargav, S.~L. Cotton, and D.~E. Simmons, ``{Secrecy Capacity Analysis over
  $\kappa$-$\mu$ Fading Channels: Theory and Applications},'' \emph{IEEE Trans.
  Commun.}, vol.~64, no.~7, pp. 3011--3024, July 2016.

\bibitem{Gomez2015}
G.~Gomez, F.~J. Lopez-Martinez, D.~Morales-Jimenez, and M.~R. McKay, ``{On the
  Equivalence Between Interference and Eavesdropping in Wireless
  Communications},'' \emph{IEEE Trans. Veh. Technol.}, vol.~64, no.~12, pp.
  5935--5940, Dec 2015.

\bibitem{Peppas2016}
K.~P. Peppas, N.~C. Sagias, and A.~Maras, ``{Physical Layer Security for
  Multiple-Antenna Systems: A Unified Approach},'' \emph{IEEE Trans. Commun.},
  vol.~64, no.~1, pp. 314--328, Jan 2016.

\bibitem{Goldsmith1997}
A.~J. Goldsmith and P.~P. Varaiya, ``Capacity of fading channels with channel
  side information,'' \emph{IEEE Trans. Inf. Theory}, vol.~43, no.~6, pp.
  1986--1992, Nov 1997.

\bibitem{Gold05}
A.~Goldsmith, \emph{Wireless communications}.\hskip 1em plus 0.5em minus
  0.4em\relax Cambridge university press, 2005.

\bibitem{Durgin2002}
G.~D. Durgin, T.~S. Rappaport, and D.~A. de~Wolf, ``{New analytical models and
  probability density functions for fading in wireless communications},''
  \emph{IEEE Trans. Comm.}, vol.~50, no.~6, pp. 1005--1015, June 2002.

\bibitem{Beckmann1962}
P.~Beckmann, ``Statistical distribution of the amplitude and phase of a
  multiply scattered field,'' \emph{JOURNAL OF RESEARCH of the National Bureau
  of Standards-D. Radio Propagation}, vol. 66D, no.~3, pp. 231--240, 1962.

\bibitem{Beckmann1964}
------, ``Rayleigh distribution and its generalizations,'' \emph{RADIO SCIENCE
  Journal of Research NBS/USNC-URSI}, vol. 68D, no.~9, pp. 927--932, September
  1964.

\bibitem{Tellambura2010}
C.~Tellambura and D.~Senaratne, ``{Accurate computation of the MGF of the
  lognormal distribution and its application to sum of lognormals},''
  \emph{IEEE Trans. Commun.}, vol.~58, no.~5, pp. 1568--1577, May 2010.

\bibitem{Gradstein2007}
\BIBentryALTinterwordspacing
I.~S. Gradshteyn and I.~M. Ryzhik, \emph{Table of Integrals, Series and
  Products}, 7th~ed.\hskip 1em plus 0.5em minus 0.4em\relax Academic Press Inc,
  2007. [Online]. Available: \url{http://www.worldcat.org/isbn/012294755X}
\BIBentrySTDinterwordspacing

\bibitem{Srivastava1985}
H.~M. Srivastava and P.~W. Karlsson, \emph{{Multiple Gaussian Hypergeometric
  Series}}.\hskip 1em plus 0.5em minus 0.4em\relax {John Wiley \& Sons}, 1985.

\bibitem{Paris2010}
J.~F. Paris, ``{Closed-form expressions for Rician shadowed cumulative
  distribution function},'' \emph{Electronics Letters}, vol.~46, no.~13, pp.
  952--953, June 2010.

\bibitem{Erdelyi1954}
A.~Erd{\'e}lyi, W.~Magnus, F.~Oberhettinger, and F.~G. Tricomi, \emph{Tables of
  integral transforms. {V}ol. {I}}.\hskip 1em plus 0.5em minus 0.4em\relax
  McGraw-Hill Book Company, Inc., New York-Toronto-London, 1954.

\bibitem{Morales2014}
D.~Morales-Jimenez, F.~J. Lopez-Martinez, E.~Martos-Naya, J.~F. Paris, and
  A.~Lozano, ``{Connections Between the Generalized Marcum Q -Function and a
  Class of Hypergeometric Functions},'' \emph{IEEE Trans. Inf. Theory},
  vol.~60, no.~2, pp. 1077--1082, Feb 2014.

\bibitem{Morales2010}
D.~Morales-Jimenez and J.~F. Paris, ``Outage probability analysis for
  $\eta$-$\mu$ fading channels,'' \emph{IEEE Commun. Lett.}, vol.~14, no.~6,
  pp. 521--523, June 2010.

\bibitem{Lopez2016}
F.~J. Lopez-Martinez, J.~F. Paris, and J.~M. Romero-Jerez, ``{The
  $\kappa$-$\mu$ Shadowed Fading Model with Integer Fading Parameters},''
  \emph{arXiv preprint arXiv:1609.00317}, 2016.

\bibitem{Maleki2014}
S.~Maleki, A.~Kalantari, S.~Chatzinotas, and B.~Ottersten, ``Power allocation
  for energy-constrained cognitive radios in the presence of an eavesdropper,''
  in \emph{Acoustics, Speech and Signal Processing (ICASSP), 2014 IEEE
  International Conference on}, May 2014, pp. 5695--5699.

\bibitem{Classen2015}
J.~Classen, J.~Chen, D.~Steinmetzer, M.~Hollick, and E.~Knightly, ``{The Spy
  Next Door: Eavesdropping on High Throughput Visible Light Communications},''
  in \emph{Proceedings of the 2Nd International Workshop on Visible Light
  Communications Systems}, 2015.

\bibitem{Ermolova2014}
N.~Y. Ermolova and O.~Tirkkonen, ``Outage probability analysis in generalized
  fading channels with co-channel interference and background noise:
  $\eta$-$\mu$/$\eta$-$\mu$, $\eta$-$\mu$/$\kappa$-$\mu$ and
  $\kappa$-$\mu$/$\eta$-$\mu$ scenarios,'' \emph{IEEE Trans. Wirel. Commun.},
  vol.~13, no.~1, pp. 291--297, Jan. 2014.

\end{thebibliography}


%
%
%



\end{document}